\documentclass[11pt]{article}
\usepackage{epsfig,enumerate,amsmath,amsfonts,amssymb,amsthm,mathrsfs,lineno,ifpdf}
\usepackage{indentfirst,relsize}
\usepackage[numbers]{natbib}
\usepackage{setspace,graphicx}
\usepackage{latexsym}
\usepackage[all]{xy}
\usepackage[usenames,dvipsnames]{pstricks}
\usepackage{pst-grad} 
\usepackage{pst-plot} 
\usepackage[margin = 3.20cm]{geometry}

\theoremstyle{plain}
\newtheorem{theo}{Theorem}
\newtheorem{lemma}[theo]{Lemma}
\newtheorem{pre}[theo]{Proposition}

\newtheorem{cl}[theo]{Claim}

\theoremstyle{definition}
\newtheorem{defi}[theo]{Definition}

\newcommand{\prob}[1]{\mathrm{Pr}\left[#1\right]}

\def \Saarbrucken {{Saarbr\"{u}cken}}
\def \Matousek {{Matou\v{s}ek}}

\def \calbddsets {{\mathcal{P}_l}}
\def \calh {{\mathcal{H}}}

\title{
Size sensitive packing number for Hamming cube\\
and its consequences
}
\author{
Kunal Dutta
\footnote{
D1: Algorithms \& Complexity, 
MPI for Informatics,
\Saarbrucken, Germany
}
\footnote{
Supported by the Indo-German Max Planck Center for Computer Science (IMPECS).
}
\and
Arijit Ghosh
\footnotemark[1]
\footnotemark[2]
}

\begin{document}

\maketitle

\begin{abstract}
We prove a size-sensitive version of Haussler's Packing 
lemma~\cite{Haussler92spherepacking} for set-systems with bounded
primal 
shatter dimension, which have an additional {\em size-sensitive
property}. This answers a 
question asked by Ezra~\cite{Ezra-sizesendisc-soda-14}. We also partially 
address another point raised by Ezra regarding overcounting of sets in her
chaining procedure. As a
consequence of these improvements, 
we get an improvement on the size-sensitive
discrepancy bounds for set systems with
the above property. Improved bounds on the discrepancy for these
special set systems also 
imply an improvement in the sizes of {\em relative 
$(\varepsilon, \delta)$-approximations} and 
$(\nu, \alpha)$-samples.\footnote{At the time of submission, we have 
become aware of a similar packing result proven simultaneously by Ezra. 
However, 
we note that our proof of the main packing lemma is quite different from 
hers. Also, the focus of our paper is on discrepancy bounds and sampling 
complexity.}
 \end{abstract}

\section{Introduction}
\label{sec-introduction}


A \emph{set system} or \emph{range space} $(X,\mathcal{S})$ is a ground set $X$ 
and a collection $\mathcal{S} \subset 2^X$ of subsets of $X$, called \emph{ranges}.

In this paper we are interested in set systems that have 
{\em bounded primal shatter dimension}.
So, let's begin 
by recalling the definition of {\em primal shatter function} which plays an important 
role in this paper:
\begin{defi}[Primal shatter function; see~\cite{matousek-geomdisc-2009}]
  The primal shatter function of a set system $(X, \mathcal{S})$ is defined as
  $$
  \pi_{\mathcal{S}}(m) = \max_{Y \subset X, \, |Y| = m} |\mathcal{S}|_{Y}|
  $$
  where $\mathcal{S}|_{Y} = \{ S\cap Y : \; S \in \mathcal{S}\}$\footnote{
    Note that for the rest of this paper we will call $\mathcal{S}|_{Y}$ the 
    {\em projection} of $Y$ onto $\mathcal{S}$.}. 
\end{defi}

A set system $(X,\mathcal{S})$ with $|X| = n$ has a 
{\em primal shatter dimension} $d$ if for all $m \leq n$
$\pi_{\mathcal{S}}(m) = O(m^{d})$. \footnote{ Strictly speaking, the primal 
shatter dimension is defined over a family $\mathcal{F}=\{\mathcal{F}_i\}_{i=1}^\infty$ of 
set systems, where for each $i$, $\mathcal{F}_i$ is a sub-family of set systems whose 
ground set has exactly $i$ elements. The constant of proportionality is common for all 
members of $\mathcal{F}$.} 

Even though it is natural to consider {\em VC-dimension} of set systems 
that arise in geometric settings, but a set system $(X,\mathcal{S})$ with VC-dimension $d$
also implies that the primal shatter dimension of $(X,\mathcal{S})$ is
$d$, 
see~\cite{matousek-geomdisc-2009}.
From this point onward we will be looking at set systems that have bounded 
primal shatter dimension.

For a set system $(X,\mathcal{S})$, a subset $\mathcal{P} \subseteq \mathcal{S}$ 
is $\delta$-separated if for all 
$S_{1}, \, S_{2} \, (\neq S_{1}) \in \mathcal{P}$ we have more than $\delta$ elements in
the set $S_{1} \Delta S_{2} = \left( S_{1} \setminus S_{2} \right) \cup \left( S_{2} \setminus S_{1}\right)$, i.e., 
{\em symmetric difference distance} $|S_{1} \Delta S_{2}|$
between $S_{1}$ and $S_{2}$ is strictly greater than $\delta$. 
A $\delta$-{\em packing} for $(X,\mathcal{S})$ is inclusion-maximal $\delta$-separated subsets of 
$\mathcal{S}$.

Let $X = [n] = \{1, \, \dots, \, n \}$ be the ground set and
$\mathcal{S}$ be a subset of $2^{X}$. Then it is natural to associate
sets $S \in \mathcal{S}$ with the $n$-dimensional {\em Hamming
  cube}, where $S$ will be mapped to the vertex $v_{S}$ whose $i$-th coordinate 
is ``one'' if $i \in S$, otherwise it is ``zero'', i.e., $v_{S}$ is
the indicator vector for the set $S$.
In this setting symmetric difference distance $|S_{1}\Delta S_{2}|$
between two set $S_{1}$ and $S_{2}$ becomes equal to the {\em Hamming 
distance} between the two vertices $v_{S_{1}}$ and $v_{S_{2}}$ in the 
Hamming cube. Therefore the problem of finding $\delta$-packing boils 
down to finding inclusion-maximal set of vertices such that for any two 
vertices the Hamming distance is greater than $\delta$.

The problem of bounding the size of a $\delta$-packing
of a set system has been an important question.
In a breakthrough paper~\cite{Haussler92spherepacking},
Haussler proved an asymptotically tight bound on the size of the largest $\delta$-packing of 
a set systems with bounded primal shatter dimension:
\begin{theo}[Haussler's packing lemma~\cite{Haussler92spherepacking, matousek-geomdisc-2009}]
  Let $d> 1$ and $C$ be constants, and let $(X, \mathcal{S})$ be a set system with $|X| = n$
  and whose primal shatter function satisfies $\pi_{\mathcal{S}}(m) \leq Cm^{d}$ for all $1\leq m \leq m$,
  i.e., primal shatter dimension $d$.  
  If $\delta$ be an integer, $1\leq \delta \leq n$, and let $\mathcal{P} \subseteq \mathcal{S}$
  be $\delta$-packed then 
  $$
  |\mathcal{P}| = O((n/\delta)^{d}).
  $$
    Note that the constant in big-$O$ depends only on $d$ and $C$. 
\label{thm-Haussler-packing-lemma}
\end{theo}
\Matousek~\cite{matousek-geomdisc-2009} remarked that Haussler's proof of the packing lemma
uses a ``probabilistic argument which looks like a magician's trick''.
Haussler actually proved the result for set systems with bounded
VC-dimension, but it was verified by 
Wernisch~\cite{Wernisch-packing-92} to also work for set
systems with bounded primal shatter dimension. 
For a proof of Theorem~\ref{thm-Haussler-packing-lemma} refer to 
\Matousek's book on Geometric Discrepancy~\cite{matousek-geomdisc-2009}, where
\Matousek~followed Chazelle's~\cite{Chazelle-packing-92} simplified proof of the above theorem.


Ezra~\cite{Ezra-sizesendisc-soda-14} 
refined the definition of primal shatter dimension to make it size sensitive. 
Specifically, for any $Y \subseteq X$ with $|Y| = m$, where $1 \leq m \leq n$,
and for any parameter $1 \leq k \leq m$, the number of sets of size at most $k$ 
in the projection 
$\mathcal{S}|_{Y}$
of $Y$ onto $\mathcal{S}$ is $Cm^{d_{1}}k^{d-d_{1}}$, where $C$ is a constant,
$d$ is the primal 
shatter dimension and $1 \leq d_{1} \leq d$.\footnote{Similar to the primal shatter 
dimension, we mention the caveat that the  size-sensitive shattering constants are 
defined for a family of set systems, where $n$ and $k$ both go to infinity, independently 
of each other.} This is a generalisation of primal 
shatter function, and for the rest of this paper we will call $d_{1}$ and $d_{2} = d-d_{1}$
size-sensitive shattering dimensions (or constants). 
Ezra ~\cite{Ezra-sizesendisc-soda-14} 
gave a bound of $O\left(\frac{j^d2^{jd}}{2^{(i-1)d_2}}\right)$, for packings of sets of 
size $O(n/2^{i-1})$ having separation $n/2^j$. 
Ezra further conjectured that the factor of $j^d$ was not essential and 
could be removed. This would have made the
bound optimal up to some constants .

The main contribution of this paper is to get the following 
a size sensitive analog of the Haussler's 
packing result:
\begin{theo}[Size sensitive packing lemma] \label{thm-main-size-sens-bd}
  Let $\delta \in [n]$, and let $\mathcal{P}$ be a $\delta$-separated
  set system having primal 
  shatter dimension $d$, and size-sensitive shattering constants $d_1$ and $d_2$. Let $\mathcal{P}_l$
  be the sets of size $l$ in $\mathcal{P}$. Let $M = M(l) = |\mathcal{P}_l |$. Then
  $$  M \leq c^*\left(\frac{n}{\delta}\right)^{d_1}\left(\frac{l}{\delta}\right)^{d_2}.$$
  where $c^*$ is independent of $n,l,\delta$.
\end{theo}
Applying the above bound to Ezra's scenario, we get $O(2^{jd}/2^{(i-1)d_2})$.
Thus, we prove that the extra polylog factors in Ezra's bound can be removed and 
answer her question in the affirmative.



\paragraph*{Combinatorial discrepancy} 

Given a set system $(X, \mathcal{S})$
where $X = [n]$, 
in {\em combinatorial discrepancy} we are interested in 
finding a {\em bi-coloring} $\chi : X \rightarrow \{-1, \, +1\}$ such that 
worst imbalance $\max_{S_{i} \in \mathcal{S}} |\chi(S)|$, where $\chi(S_{i}) = \sum_{j \in S_{i}} \chi(j)$,
in the set system is minimised. The discrepancy of $(X, \mathcal{S})$ is defined as 
$$
\mathrm{disc}(\mathcal{S}) = \min_{\chi} \max_{S\in \mathcal{S}} |\chi(S)|.
$$

Using {\em partial coloring or entropy method} of
Beck~\cite{Beck-discrepancy-integer-sequence-81} 
and an innovative {\em chaining method} 
(originally due to Kolmogorov)
to get a decomposition of sets in the range space,  
Matou$\mathrm{\check{s}}$ek~\cite{Matousek-tight-half-spaces-95,matousek-geomdisc-2009} 
proved an important result for the case 
of set system with bounded primal shatter dimension:

\begin{theo}[\cite{Matousek-tight-half-spaces-95,matousek-geomdisc-2009}]\label{thm-Matousek-discrepancy-primal-shatter} 
  Let $d > 1$ be a constant, and let $(X, \mathcal{S})$ be a set system 
  with $\pi_{m}(\mathcal{S}) \leq C m^{d}$, where $C > 0$ and $d > 1$
  are constants. Then $\mathrm{disc}(\mathcal{S}) = O(n^{1/2 - 1/2d})$,
  where the constant is big-$O$ depends only on $d$ and $C$.
  \label{thm-Matousek-discrepancy-bounded-primal-dimension}
\end{theo}

Ezra~\cite{Ezra-sizesendisc-soda-14} 
generalised the above result to be case of set systems
with size sensitive primal shatter dimensions $d_{1}$ and $d_{2}$
and also make the discrepancy dependent on the size of the sets:

\begin{theo}[\cite{Ezra-sizesendisc-soda-14}]\label{thm-Ezra-size-sensitive-discrepancy}
  Let $(X, \mathcal{S})$ be a finite set system of primal shatter
  dimension $d$ with the additional property that in any set system
  restricted to $Y \subseteq X$, the number of sets of size $k \leq
  |Y|$
  is $O(|Y|^{d_{1}} k^{d-d_{1}})$, where $1\leq d_{1} \leq d$. Then
  $$ 
  \mathrm{disc}(\mathcal{S}) = \left\{ \begin{array}{ll} 
      O\left(|S|^{d_2/(2d)}n^{(d_1-1)/(2d)}\log^{1/2+1/2d}n \right), & \mbox{ if } d_1>1 \\
      O\left(|S|^{d_{2}/(2d)}\log^{3/2 + 1/2d} n \right), & \mbox{ if } d_1=1 
    \end{array} \right.
  $$
  where $d_{2} = d-d_{1}$.
\end{theo}

This bound is slightly suboptimal for the case for points and
halfspaces in the plane. 
Har-Peled and
Sharir.~\cite{Har-PeledS11-relative-approximation-geometry} 
proved that for the case of points and
halfspaces in the plane the discrepancy bound for a set $S$ is 
$O(|S|^{1/4} \log n)$, but the bound in
Theorem~\ref{thm-Ezra-size-sensitive-discrepancy} is a considerable 
improvement for the case of points and halfspaces in three dimensional 
space obtained by Sharir and Zaban~\cite{SharirZ-range-searching-13}, which extended the construction of 
Har-Peled and Sharir~\cite{Har-PeledS11-relative-approximation-geometry}.

Using our new size sensitive packing bound and 
Ezra's ~\cite{Ezra-sizesendisc-soda-14} 
refinement of 
Matou$\mathrm{\check{s}}$ek's chaining 
trick~\cite{Matousek-tight-half-spaces-95,matousek-geomdisc-2009}, 
we get the
following improvement 
to Theorem~\ref {thm-Ezra-size-sensitive-discrepancy}.

\begin{theo}\label{thm-appln-disc-bds}
   Let $(X,\mathcal{S})$ be a (finite) set system of primal shatter dimension $d$ and size-sensitive 
   shatter constants $d_1\leq d$ and $d_2=d-d_1$. Then 
   $$ 
   \mathrm{disc}(\mathcal{S}) = \left\{ \begin{array}{ll} 
       O\left(|S|^{d_2/(2d)}n^{(d_1-1)/(2d)} f(|S|, n)\right), & \mbox{ if } d_1>1 \\
       O\left(|S|^{d_2/(2d)} f(|S|,n) \log n \right), & \mbox{ if } d_1=1 
     \end{array} \right.
   $$
   where $f(|S|,n) = \sqrt{1+2\log \left( 1+ \log \min \left\{ |S|, \,
       \frac{n}{|S|}\right\}\right)}$.
 A coloring with the above discrepancy bounds 
 can be computed in expected polynomial time. 
 \end{theo}
 
Note that the proof of both Theorems~\ref{thm-Ezra-size-sensitive-discrepancy} and 
\ref{thm-appln-disc-bds}
fundamentally uses the recent improvement to Beck's 
entropy method~\cite{Beck-discrepancy-integer-sequence-81} by 
Lovett and Meka~\cite{Lovett-Meka-discmin-focs-12}.

A set system $(X, \mathcal{S})$ is called {\em low degree}
if for all $j \in X$,
$j$ appears in at most $t \leq n$ sets in $\mathcal{S}$.
Beck and Fiala has been conjectured that discrepancy of low degree set system 
is $O(\sqrt{t})$~\cite{Beck-Fiala-integer-making-1981}. 
This is know as the {\em Beck-Fiala conjecture}.
Beck and Fiala ~\cite{Beck-Fiala-integer-making-1981}, 
using a linear programming approached showed that 
discrepancy of low degree set systems is bounded by $2t-1$.
Using entropy method, one can obtain a constructive bound 
of $O(\sqrt{t} \log n)$~\cite{Bansal-discrepancy-minimization-focs-10, 
Lovett-Meka-discmin-focs-12}. See
also~\cite{Srinivasan-soda-Matrices97}.
Currently the best bound (non-constructive) is by Banaszczyk~\cite{Banaszczyk-rsa-98}, who proved 
that the discrepancy is bounded by $O(\sqrt{t\log n})$.
We prove, in Section~\ref{sec-well-behaved-Beck-Fiala}, that for low degree set systems with
$d_{1} =1$, the discrepancy is bounded by
$O(t^{1/2-1/2d}\sqrt{\log\log t}\log n)$.
Specifically, for the case of points and halfspaces in $2$-dimensional
setting, we get $O(t^{1/4} \sqrt{\log \log t} \log n)$.
Also note that our result is constructive, i.e., in expected
polynomial time we can find a coloring that matches the above
discrepancy bound. 


\paragraph*{Relative $(\varepsilon, \delta)$-approximation and
  $(\nu,\alpha)$}

Using the improved size sensitive discrepancy bounds for sets we will 
be able to improve on the previous bounds for {\em relative $(\varepsilon,
  \delta)$-approximation}
and {\em $(\nu, \alpha)$-sample}. 

For a finite set system $(X,\mathcal{S})$, we define for $S \in \mathcal{S}$
 $\overline{X}(S) = \frac{|S\cap X|}{X}$. For a given $0 < \varepsilon
 < 1$ and $0< \delta < 1$, a subset $Z \subseteq X$ is a 
{\em relative $(\varepsilon, \delta)$-approximation } if  $\forall \,
S \in \mathcal{S}$
\begin{eqnarray*}
  &\overline{X}(S)(1-\delta) \leq \overline{Z}(S) \leq
  \overline{X}(S)(1+\delta), \;\; \mbox{if}~ \overline{X}(S) \geq
  \varepsilon, \; \mbox{and}&\\
  &\overline{X}(S)-\delta \varepsilon \leq \overline{Z}(S) \leq
  \overline{X}(S)+\delta \varepsilon, \; \mbox{otherwise}&
\end{eqnarray*}
Har-Peled and
Sharir~\cite{Har-PeledS11-relative-approximation-geometry} 
showed that the notion of relative $(\varepsilon, \delta)$-approximation
and $(\nu, \alpha)$-sample are equivalent if $\nu$ is proportional $\varepsilon$ and $\alpha$
is proportional to $\delta$. 
A $(\nu, \alpha)$-sample of a set system $(X, \mathcal{S})$ is a
subset $Z \subseteq X$ satisfying the following inequality 
$\forall \, S \in \mathcal{S}$:
\begin{equation}
  d_{\nu}(\overline{X}(S), \overline{Z}(S)) : = \frac{|\overline{X}(S)
    - \overline{Z}(S)|}{\overline{X}(S) + \overline{Z}(S) + \nu} < \alpha.
\end{equation}
Relative $(\varepsilon, \delta)$-approximation is an important tool
to tackle 
problems in {\em approximate range counting}~\cite{Har-PeledS11-relative-approximation-geometry}.

In Section~\ref{sec-improving-sample-bound-via-discrepancy}, 
we prove the following bound on the size of
relative $(\varepsilon, \delta)$-approximation
which is an 
improvement over the previous 
bounds~\cite{Har-PeledS11-relative-approximation-geometry,
Ezra-smallsizeapprangespace-socg-13,
Ezra-sizesendisc-soda-14}, the most recent one being \cite{Ezra-sizesendisc-soda-14},
who gave a bound of 
$$ \begin{array}{ll}
    \max\left\{O(\log n), O\left(\frac{\log 1/(\varepsilon\delta)}
                          {\varepsilon^{\frac{d+d_1}{d+1}}\delta^{\frac{2d}{d+1}}}\right)\right\},
                          \\ & \mbox{ for } d_1 > 1, \mbox{ and } \\
    \max\left\{O(\log^{\frac{3d+1}{d+1}} n), O\left(\frac{\log^{\frac{3d+1}{d+1}} 1/(\varepsilon\delta)}
                          {\varepsilon\delta^{\frac{2d}{d+1}}}\right)\right\},
                          & \mbox{ for } d_1 = 1.
   \end{array}
.$$
. 
\begin{theo}\label{thm-main-sample-size-bound}
  Let $(X, \mathcal{S})$ be a set system with $|X| = n$, 
  primal shatter dimension $d$ and size sensitive shattering constants
  $d_{1}$ and $d_{2} = d - d_{1}$. Then for $0 < \varepsilon < 1$
  and $0< \delta < 1$, $(X, \mathcal{S})$ has a relative 
  $(\varepsilon, \delta)$-approximation of size
  $$
  O\left( \frac{\log \log^{\frac{2d}{d+1}}
      \frac{1}{\varepsilon\delta}}{\varepsilon^{\frac{d+d_{1}}{d+1}}
      \delta^{\frac{2d}{d+1}}}\right).
  $$
  for $d_{1} > 1$, and 
  $$
  \max \left\{ O\left( \log^{\frac{2d}{d+1}} n \right), 
    \, O\left( \frac{\log^{\frac{2d}{d+1}} \frac{1}{\varepsilon\delta} \, \log
        \log^{\frac{2d}{d+1}} \frac{1}{\varepsilon\delta}}
      {\varepsilon \delta^{\frac{2d}{d+1}}}\right) \right\}
  $$
  for $d_{1} = 1$. The constant in big-$O$ depends only on $d$, and 
  a relative $(\varepsilon, \delta)$-approximation with above bounds 
  can be computed in expected polynomial time. 
\end{theo}

\section{Preliminaries}
   In this section we cover some basic concepts of discrepancy theory, especially in 
a geometric setting, which will be needed in the following sections.

\subsection{The Beck-Spencer method, and the Lovett-Meka algorithm}
In the discrepancy upper bound problem, given a universe of elements $X$, and a 
subset $\mathcal{S}$ of its power set, $\mathcal{S} \subset 2^X$, we wish to find 
a coloring $\chi:X\rightarrow [-1,1]$ which minimizes the \emph{imbalanace} in every 
set $S\in \mathcal{S}$.
In~\cite{Beck-discrepancy-integer-sequence-81}, Beck introduced a technique to obtain such colorings - the method of \emph{partial 
coloring}. The idea is to color a substantial fraction of elements, while leaving others 
uncolored. This allows for low discrepancy in the partially colored universe. The remaining 
elements are then colored recursively. The technique was then further developed by 
Spencer~\cite{Spencer-six-standard-85}, 
and is one of the major techniques used extensively in discrepancy theory.
In the proofs of Beck and Spencer, the partial coloring was an existential result, and did 
not yield polynomial time algorithms to give low-discrepancy colorings.
Recently however, a breakthrough result of Bansal~\cite{Bansal-discrepancy-minimization-focs-10} 
provided the first polynomial time
algorithm to obtain low-discrepancy colorings whose existence was implied by the Beck-Spencer 
technique. Subsequently, Lovett and Meka~\cite{Lovett-Meka-discmin-focs-12} 
gave a constructive version of the 
Beck-Spencer partial coloring lemma. We describe their lemma below:

Given a parameter $\delta \geq 0$, a \emph{partial coloring} of $X$ is a function 
$\chi:X\rightarrow [-1,1]$, where if for some $x$, $|\chi(x)| \geq 1-\delta$, then we say that 
$x$ is \emph{colored}, otherwise $x$ is \emph{uncolored}.

\begin{lemma}[Lovett-Meka~\cite{Lovett-Meka-discmin-focs-12}] \label{lemma:lovett-meka}
   Let $(X,\mathcal{S})$ be a set system with $|X|=n$. Let $\Delta:\mathcal{S}\rightarrow \mathbb{R}_+$ 
be such that 
   $$ \sum_{S \in \mathcal{S}} exp(-\Delta_S^2/(16|S|)) \:\: \leq\:\:  n/16 .$$
Then, there exists $\chi:X\rightarrow [-1,1]^n$ with $|\{i:|\chi_i| =1\}|\geq n/2$, such that 
$|sum_{i\in S}\chi_i|\leq \Delta_S + 1/poly(n)$ for every $S \in \mathcal{S}$. Further, there exists 
a randomized $poly(|\mathcal{S}|, n)$-time algorithm to find $\chi$.
\end{lemma}

The above lemma can be recursively applied on the remaining uncolored elements of $X$, to obtain 
a full coloring function $\chi:X\rightarrow [-1,-1+\delta]\cup [1-\delta,1]$. This can be rounded to 
a coloring in $[-1,1]^n$ by choosing $\delta$ sufficiently small. We shall refer to each 
such application of Lemma \ref{lemma:lovett-meka} as one round of the Lovett-Meka algorithm; a complete
coloring, then, requires $O(\log n)$ such rounds. Denoting the bound for the set $S$ in the $j$-th round
by $\Delta_{S,j}$, the discrepancy $|\sum_{i\in S}\chi_i|$ of the set $S$ in the final coloring is bounded 
from above by $\Delta_S=\sum_j \Delta_{S,j}$.

\subsection{Chaining, and size-sensitive shattering constants}

Now we describe the chaining decomposition, as used by 
\Matousek~\cite{Matousek-tight-half-spaces-95,matousek-geomdisc-2009}, 
and further refined by Ezra~\cite{Ezra-sizesendisc-soda-14}. 
We are given a set system $(X,\mathcal{S})$. For each 
$j=0,\ldots,\log n=k$, we first form a \emph{maximal} $\delta=n/2^j$-separated family, 
$\mathcal{F}_j \subset 2^X$. Clearly, $\mathcal{F}_k = \mathcal{S}$, and $\mathcal{F}_0=\emptyset$.
Since each family $\mathcal{F}_i$ is maximal, for each $F_i \in \mathcal{F}_i$, there exists a 
$F_{i-1}\in \mathcal{F}_{i-1}$, such that $|F_i \Delta F_{i-1}| \leq n/2^{i-1}$. If $F_i \in \mathcal{F}_{i-1}$, 
then we are done. Otherwise, suppose the statement were not true, 
then $F_i$ would have symmetric difference at least $n/2^{i-1}$ from every member of $\mathcal{F}_{i-1}$,
and so would have to be a member of $\mathcal{F}_{i-1}$, which contradicts the maximality of $\mathcal{F}_{i-1}$.

\paragraph*{The first decomposition} 
Notice that since every $S \in \mathcal{S}$ lies in $\mathcal{F}_k$, applying 
the above property, we can find for each such $S$, $F_{k-1} \in \mathcal{F}_{k-1}$ such that the Hamming distance
of $S$, $F_{k-1}$ is at most $n/2^{k-1}=2$. Let $A_{k}=S\setminus F_{k-1}$, and $B_{k}= F_{k-1}\setminus S$, 
and let $F_k$ denote $S$. 
Clearly, $S= F_k = (F_{k-1} \cup A_{k-1}) \setminus B_{k-1}$.
Extending this argument further, we get that 
\begin{eqnarray*}
S = F_k &=&   (\ldots ((((F_0=\emptyset \cup A_1) \setminus B_1)\cup A_2)\setminus B_2)\cup \ldots \cup A_{k}) \setminus B_{k} \\
&=&   (\ldots (((A_1 \setminus B_1)\cup A_2)\setminus B_2)\cup \ldots \cup A_{k}) \setminus B_{k}
\end{eqnarray*}
where for each $j=1,\ldots,k$, $F_{j-1}\in \mathcal{F}_{j-1}$ is the closest neighbour of $F_j$ in $\mathcal{F}_{j-1}$, and 
$A_i := F_i \setminus F_{j-1}$, and $B_{j} := F_{j-1} \setminus F_i$.
We call the sequence $S= F_k \rightarrow F_{k-1} \rightarrow ... F_1$ as the \emph{closest-neighbour chain} of 
$S$.

This decomposition is now made size-sensitive by the following refinement: partition the 
sets in $\mathcal{S}$ into $\mathcal{S}_1,\ldots,\mathcal{S}_k$, where for $S\in \mathcal{S}$,
$S\in \mathcal{S}_i$ if and only if
$$ \frac{n}{2^{i}} \leq S \leq \frac{n}{2^{i-1}}.$$
For a fixed $S_i \in \mathcal{S}_i$, consider the truncated closest-neighbour chain and the corresponding decomposition:
$$S_i= F_k^i \rightarrow F_{k-1}^i \rightarrow ... F_{i-1}^i,$$
$$S =   (\ldots ((F_{i-1}\cup A_i)\setminus B_i)\cup \ldots \cup A_{k}) \setminus B_{k} .$$
We now construct the size-sensitive families $\mathcal{F}_j^i$, by following, for each $S \in \mathcal{S}_i$, the 
truncated closest-neighbour chain of $S$, and assigning each $F_j^i$ in this chain to the family $\mathcal{F}_j^i$.
The following properties can be easily proven using the triangle-inequality on the closest-neighbour chain of $S_i$:
\begin{pre}
    For each for the sets $F_j^i \in \mathcal{F}_j^i$, $j = i-1,\ldots,k$, we have: 
    $$ |S\Delta F_j^i| < O\left(\frac{n}{2^{j-1}}\right) .$$
\end{pre}

\begin{pre}
    For each for the sets $F_j^i \in \mathcal{F}_j^i$, $j = i-1,\ldots,k$, we have: 
    $$ |F_j^i| < O\left(\frac{n}{2^{j-1}}\right) .$$
\end{pre}

Similar to the definitions of $A_j, B_j$, we construct the size-sensitive families 
$A_j^i= F_j^i \setminus F_{j-1}^i$, and $B_j^i = F_{j-1}^i \setminus F_j^i$, for each 
$F_j^i\in \mathcal{F}_j^i$.
Finally, for each fixed $i = 1,\ldots, k$, $j=i-1,\ldots, k$, let $\mathcal{M}_j^i$ 
denote the collection of $A_j^i, B_j^i$.

Observe that for each $i,j$, we have $|\mathcal{M}_j^i|=2|\mathcal{F}_j^i|$. Further, 
the size of each $A_j^i,B_j^i\in \mathcal{M}_j^i$, is $O\left(n/2^{i-1}\right)$.
We shall later apply the Beck-Spencer partial coloring technique, on the set-system $\left(X,\bigcup_{i,j}\mathcal{M}_j^i\right)$.

\section{Size-sensitive packing bound
}
%
%

  In this section, we shall prove a size-sensitive version of Haussler's upper bound for $\delta$-
  separated systems in set-systems of bounded primal shatter dimension.
  By Haussler's result \cite{Haussler92spherepacking}, we know that $M = O(n/\delta)^d=
  (n/\delta)^{d_1}(l/\delta)^{d_2}.g(n,l,\delta)^d$, where $g(n,l,d) = O((n/l)^{d_2})$.
  We want to show the optimum upper bound for $g$ is independent of $n,l$.
  We shall build on Chazelle's presentation of Haussler's proof, (which has been described by 
  {\Matousek} as ``a magician's trick") as explained in \cite{matousek-geomdisc-2009}. 
  We shall show that the optimal bound (up to constants) is in fact, $g=c^*$, where $c^*$ is the fixed point of 
  $f(x)= c'\log x$, with $c'>0$ independent of $n,l,\delta$. 

  \paragraph*{Intuition} 

  We provide some intuition for our extension of Haussler's proof below (at least to the reader familiar 
  with it). A na\"{i}ve attempt to extend Haussler's proof to 
  size-sensitive shattering constants fails because the proof essentially uses a random 
  sampling set  $A$, which for our purposes, has to behave somewhat like a $(\varepsilon,\delta/n)$-approximation, 
  (at least with respect to upper-bounds on the intersection sizes), with at least a constant probability. This 
  is unlikely to be true for a random set $A$. We shall therefore, not require that $A$ behave like an $(\varepsilon, \delta/n)$-
  approximation. Instead, we shall allow some sets in $\calbddsets$ to have larger than expected intersections
  with $A$, and control the expected number of such `bad' sets by our choice of the size of the sample set $A$.

  \paragraph*{Details}
  
  Let $A \subset X$ be a random set, constructed by choosing each element $u \in X$ randomly with 
  probability $p = \frac{36dK}{\delta}$, where $K\geq 1$ is a parameter to be 
  fixed later. Let $s:= |A|$. Define $\calh = \calbddsets |_A$.
  Consider the unit distance graph $UD(\calh)$. For each set $Q \in \calh$, define the weight of $Q$ as:
  $$ 
  w(Q) = \#\{S \in \calbddsets: S \cap \calh = Q \}.
  $$
  Observe that $ \sum_{Q \in \calh} w(Q) = M.$

  Let $E = E(UD(\calh))$, the edge set of $UD(\calh)$. Now define the weight of an \emph{edge} 
  $e = (Q,Q') \in E$ as 
  $$ w(e) = \min (w(Q),w(Q')).$$
  Let $W := \sum_{e \in E} w(e) .$
  We claim that
  \begin{cl}
     For any $A\subset X$,  
     $$ W \leq 2d\sum_{Q \in \calh} w(Q) = 2dM.$$
  \end{cl}

  \begin{proof}
The proof is based on the following lemma, proved by Haussler~\cite{Haussler92spherepacking} for 
set systems with bounded VC dimension, and later verified by Wernisch~\cite{Wernisch-packing-92} to also work for set
systems with bounded primal shatter dimension.
The following version appears in \Matousek's book 
on Geometric Discrepancy~\cite{matousek-geomdisc-2009}:

\begin{lemma}[\cite{Haussler92spherepacking}, \cite{matousek-geomdisc-2009}]
   Let $\mathcal{S}$ be a set-system of primal shatter dimension $d$ on a 
finite set $X$. Then the unit-distance graph $UD(\mathcal{S})$ has at most 
$d|V(\mathcal{S})|$ edges.
\end{lemma}

  Let $\mathcal{S}$ be $\calh$ . Since $\calh$ has primal shatter dimension $d$, the 
  lemma implies that there exists a vertex $v \in V(\calh)$, whose degree is at most 
  $2d$. Removing $v$, the total vertex weight drops by $w(v)$, and the total edge weight 
  drops by at most $2dw(v)$. Continuing the argument until all vertices are removed, 
  we get the claim.
  \end{proof}

  Next, we shall prove a lower bound on the expectation $E[W]$.
  Choose a random element $a \in A$. Let $A' := A \setminus \{a\}$. Note that $A'$ 
  is a random subset of $X$, chosen with probability $p' = p-1/n$. 
  Crucially, one 
  can consider the above process equivalent to first choosing $A'$ by selecting each
  element of $X$ with probability $p'$, and then selecting a uniformly random element 
  $a\in X\setminus A'$ with probability $1/n$.

  Let $E_1 \subset E$ be those edges $(Q,Q')$ of $E$ for which $Q \Delta Q' = \{a\}$, 
  and let $W_1 = \sum_{e \in E_1} w(e) .$
  We need to lower bound $E[W_1]$.
  Given $A' \subset X$, let $Y = Y(A') := \#\{S \in \calbddsets : |S\cap A'|> c(l/\delta)\}$ i.e. 
  the number of sets in $\calbddsets$, each of whose intersection with $A'$ has more than $c(l/\delta)$ elements, 
  (where $c$ shall be chosen appropriately). Let $Nice$ denote the event 
  $(Y \leq 8E[Y]) \wedge( np/2 \leq s \leq 3np/2) = N_Y \wedge N_S$. 
  Conditioning $W$ on $Nice$, we get:
  \begin{eqnarray*}
  E[W] &=&   \prob{Nice}E[W|Nice] + \prob{\overline{Nice}}E[W|\overline{Nice}]\\
  &>& \prob{Nice}E[W|Nice] 
  \end{eqnarray*}

  By Markov's inequality: 
  $$\prob{\bar{N_Y}} = \prob{Y\leq 8E[Y]} \leq 1/8,$$ 
  and using Chernoff bounds, 
  $$\prob{\bar{N_S}} = \prob{|s-np|> np/2} \leq 2e^{-36dKn/(3.2^2\delta)} 
                     << 1/4,$$ 
  since $n/\delta\geq 1$. We get that $\prob{Nice}\geq 7/8-e^{-4dK} \geq 3/4$ 
  for $dK\geq 1$.
  Hence, $$E[W] \geq (3/4)E[W|Nice] \geq \frac{3(np/2)}{4}E[W_1|Nice],$$
  where the last inequality follows by symmetry of the choice of $a$ from $A$, 
  and the lower bound on $s$ when the event $Nice$ holds.

  Hence, $E[W]\geq (3np/8)E[W_1|Nice]$. So to lower bound $E[W_1]$
  up to constants, it suffices just to lower bound $E[W_1|Nice]$.
  Let $W_2$ denote $W_1|Nice$. Consider now $E[W_2|A']$. That is, consider 
  a fixed set $A'$ whose size is between $np/2$ and $3np/2$, and which is 
  such that 
  the number of sets $S \in \calbddsets$ which intersect $A'$ in more 
  than $cl/\delta$ vertices, is at most $8E[Y]$. We shall lower bound 
  $W_1$ for this choice of $A'$.
  
  By definition, $W_1 = \sum_{e \in E_1} w(e)$. Consider the equivalence 
  classes of $\calbddsets$ formed by their intersections with $A'$:
  $$ \calbddsets = \mathcal{P}_1 \cup \mathcal{P}_2 \cup \ldots \cup 
     \mathcal{P}_r .$$
  Define $Bad \subset [r]$ to be those indices $j$ for which 
  $\mathcal{P}_j$ is such that 
  $$\forall S \in \mathcal{P}_j: | S \cap A'| > 8c(l/\delta).$$
  Further, let $Good$ be $[r]\setminus Bad$.
  Since $Nice$ holds, we have:
  $$ \sum_{j \in Bad} | \mathcal{P}_j| \leq 8E[Y].$$
  Consider a class $\mathcal{P}_i$ such that $i \in Good$.
  Let $P_1 \subset \mathcal{P}_i$ be those sets in $\mathcal{P}_i$ which 
  contain $a$, 
  and let $P_2 = \mathcal{P}_i \setminus P_1$. Let $b = |\mathcal{P}_i|$,  
  $b_1 = |P_1 |$ 
  and $b_2 = |P_2 |$.
  Then the edge $e\in E_1$ formed by the projection of 
  $\mathcal{P}_i$ in $A$, has weight 
  $$w(e) = \min (b_1,b_2) \geq \frac{b_1b_2}{b}.$$ 
  For a given ordered pair of sets $S,S' \in \mathcal{P}_i$, the probability 
  that $a \in S \Delta S'$ is 
  $\frac{\delta}{n-|A'|}$, which is at least $\frac{\delta}{n}$. Therefore, 
  the expected weight of $e$ (conditioned on $Nice$ and $A'$) is at least:
  $$E[w(e)|Nice\cap A'] \geq \frac{b(b-1)}{b}.\frac{\delta}{n} = 
   (b-1)\frac{\delta}{n} = (|\mathcal{P}_i|-1)\frac{\delta}{n} .$$
  Hence, the expected weight of $E_1$ is:
  \begin{eqnarray*}
    E[W_2|A'] \geq \sum_{e\in E_1} w(e) \geq \sum_{i \in Good} 
    (|\mathcal{P}_i|-1)\frac{\delta}{n} 
  \end{eqnarray*}
  But by the size-sensitive shattering property, we have that 
  $$\forall j \in Good, \:\: |(\mathcal{P}_j|_{A'})| \leq Cs^{d_1}(clp)^{d_2}.$$
  Substituting in the lower bound for $E[W_2]$, we get:
  \begin{eqnarray*}
  E[W_2|A'] &\geq& \left(\left(\sum_{i \in Good} |\mathcal{P}_i|\right)
           -C(2np)^{d_1}(clp)^{d_2}\right)\frac{\delta}{n} \\
  &\geq& \left(|\mathcal{P}_l| - 8E[Y]
           -C(6dK)^d.2^{d_1}c^{d_2}\left(\frac{n}{\delta}\right)^{d_1}\left(\frac{l}{\delta}\right)^{d_2}\right)\frac{\delta}{n}  \\
  &\geq& \left(M - 8E[Y]
           -C_1K^d\left(\frac{n}{\delta}\right)^{d_1}\left(\frac{l}{\delta}\right)^{d_2}\right)\frac{\delta}{n} 
  \end{eqnarray*}
  where $C_1 = C.(6d)^d2^{d_1}c^{d_2}$.
  Since the above holds for each $A'$ which satisfies $Nice$, we get that 
  $$E[W_2] \geq \left(M - 8E[Y]
           -C_1K^d\left(\frac{n}{\delta}\right)^{d_1}\left(\frac{l}{\delta}\right)^{d_2}\right)\frac{\delta}{n},$$
  Comparing with the upper bound on $W$,
  $$ (3np/8)E[W_1|Nice] \leq  E[W] \leq 2dM ,$$
  and substituting the lower bound $E[W_1|Nice]$, and solving for $M$, we get 
  $$M \leq \frac{(27K/4)\left(8E[Y] + C_1K^d
           \left(\frac{n}{\delta}\right)^{d_1}\left(\frac{l}{\delta}\right)^{d_2}\right)}{(27K/4-1)} .$$
  The following claim therefore, completes the proof:
  \begin{cl} \label{clm:upbddexp}
      For $K= \max\{1,(\ln g)/36\}$,
      $$E[Y] \leq C_2\left(\frac{n}{\delta}\right)^{d_1}
             \left(\frac{l}{\delta}\right)^{d_2}.$$
  \end{cl}
  Indeed, substituting the choice of $K$ and the value of $E[Y]$ from 
  Claim \ref{clm:upbddexp}, we get that 
 \begin{eqnarray*}
 g^d\left(\frac{n}{\delta}\right)^{d_1}\left(\frac{l}{\delta}\right)^{d_2}  = M
   &\leq& \frac{C_1K^d\left(\frac{n}{\delta}\right)^{d_1}
    \left(\frac{l}{\delta}\right)^{d_2}
     +  8C_2\left(\frac{n}{\delta}\right)^{d_1}
            \left(\frac{l}{\delta}\right)^{d_2}}{1-4/27K} \\
   &\leq& C_3K^d\left(\frac{n}{\delta}\right)^{d_1}\left(\frac{l}{\delta}\right)^{d_2} \\
   &\leq& C_4(\max\{1,\log g\})^d\left(\frac{n}{\delta}\right)^{d_1}\left(\frac{l}{\delta}\right)^{d_2} 
 \end{eqnarray*}
 This implies that $g \leq C_4\max\{1,\log g\}$. Since for any non-negative 
 $g$, we have $g \geq C_4\log g$, therefore, it suffices to take $g\leq C_4$,
 i.e. $g=c^*$, where $c^*$ is a constant independent of $n,l,\delta$.
      
It only remains to prove Claim \ref{clm:upbddexp}:
\begin{proof}[Proof of Claim \ref{clm:upbddexp}]
  The proof follows easily from Chernoff bounds. Fix $S \in \calbddsets$. Let $Z=|A' \cap S|$. 
  Then $E[Z]=|S|p'=lp'$. Since $A'$ is a random set chosen with probability $p' = p -1/n$, the probability that 
  $Z\geq clp' = 36cdKl/\delta-cl/n$ is upper bounded using Chernoff bounds, as:
  $$\prob{Z-E[Z]>(c-1)E[Z] } \leq e^{-E[Z]} \leq e^{-36dKl/\delta},$$
  for $c= 1.01e$ and $n\geq 100$, say. Hence the expected number $E[Y]$ of sets, each of which intersect $A'$ in more 
  than $36cdKl/\delta$ elements, is at most:
  $$ E[Y] \leq M.e^{-36dKl/\delta} \leq Me^{-36dK},$$
  since $l\geq \delta$. Substituting the value of $M$ and also $K$ in terms of $f$, we have 
  \begin{eqnarray*}
  E[Y] \leq g^d\left(\frac{n}{\delta}\right)^{d_1}\left(\frac{l}{\delta}\right)^{d_2}e^{-36dK} 
     \leq \left(\frac{n}{\delta}\right)^{d_1}\left(\frac{l}{\delta}\right)^{d_2}e^{d(\ln g-36K)}
     \leq \left(\frac{n}{\delta}\right)^{d_1}\left(\frac{l}{\delta}\right)^{d_2} 
  \end{eqnarray*}
  for $K\geq (\ln g)/36$.
\end{proof}

This completes the proof of Theorem~\ref{thm-main-size-sens-bd}.

\section{Size-sensitive discrepancy bounds}

   Now we shall use the framework of Ezra \cite{Ezra-sizesendisc-soda-14}, together with the result proved in the previous section, to obtain improved 
bounds on the discrepancy of set systems having bounded primal shatter dimension and size-sensitive shattering constants $d_1$ 
and $d_2$. Such set systems are often encountered in geometric settings, e.g. points and half-spaces in $d$ dimensions. In addition, we shall also 
derive an easy corollary for the Beck-Fiala setting, i.e. where each element of the universe $X$ has bounded degree.
In order to keep the exposition as self-contained as possible, we briefly describe Ezra's framework first.

The basic idea, as in \cite{matousek-geomdisc-2009}, 
and \cite{Ezra-sizesendisc-soda-14}, will be to consider the set system formed by 
$A_j^i$ and $B_j^i$ for each $i,j \in \{1,\ldots,k\}$, and bound the discrepancy for this system. We shall then sum the 
discrepancies of the components of each set in the original system, to obtain the total discrepancy. 

For each $i$ and $j$, define $\mathcal{M}_j^i$ as the collection of the sets $A_j^i$, $B_j^i$. In each iteration of the 
Lovett-Meka algorithm, we set a common discrepancy bound 
$\Delta_j^i$ for all the sets in $\mathcal{M}_j^i$. Note that the construction of $A_j^i$ and $B_j^i$ implies that for each 
$i,j$, $|A_j^i|=|B_j^i|=O(n/2^{j-1})$. By Theorem~\ref{thm-main-size-sens-bd}, we have that for each $i$ and $j$, 
$|\mathcal{M}_j^i|= O\left(\frac{2^{jd}}{2^{(i-1)d_2}}\right)$. Further, note that for a fixed $i$, by the construction, 
the size of each original set $S \in \mathcal{S}$ is $O(n/2^{i-1})$.
Grouping the sets having the same discrepancy parameter $\Delta_j^i$, together, we see that we need:
$$ \sum_{i=1}^k \sum_{j=i-1}^k C.\frac{2^{jd}}{2^{d_2(i-1)}} . exp\left(-\frac{(\Delta_j^i)^2}{16s_j}\right) \leq \frac{n}{16},$$
where $k:= \log n$, and $s_j=n/2^{j-1}$.
Define
$$ j_0 := \frac{\log n + d_2(i-1)}{d} - B,$$
where $B$ is a suitable constant to be set later. 
Proceeding as in~\cite{Ezra-sizesendisc-soda-14}, we shall split the sum into two parts: $j> j_0$ and $j\leq j_0$. 
Set
$$ \Delta_j^i := A.\frac{1}{(1+ |j-j_0|)^2}\left(\frac{n^{1/2-1/(2d)}}{2^{(i-1).(d_2/(2d))}}\right).\sqrt{1+2\log h},$$
where $h=(k/2-|i-k/2|)$.

For the case when $j> j_0$, let $j = j_0+r$.
\begin{eqnarray*}
\sum_{i=1}^k \sum_{j>j_0}^k C.\frac{2^{jd}}{2^{d_2(i-1)}} . exp\left(-\frac{(\Delta_j^i)^2}{16s_j}\right) &\leq& 
\sum_{i=1}^{k} \sum_{r=1}^{k-j_0} C.\frac{n2^{rd}}{2^{dB}} . exp\left(-\frac{A^2 2^{r-(B+1)}(1+2\log h)}{16(1+r)^4}\right) \\ 
&\leq& 
\sum_{i=1}^{k} C.\frac{n}{2^{dB}h^2} \sum_{r=1}^{k-j_0} 2^{rd}. exp\left(-\frac{A^2 2^{r-(B+1)}}{16(1+r)^4}\right) \\ 
&\leq& 
\sum_{i=1}^{k} C.\frac{n}{2^{dB}h^2} \sum_{r=1}^{k-j_0} exp\left(rd\ln 2 -\frac{A^2 2^{r-(B+1)}}{16(1+r)^4}\right) 
\end{eqnarray*}
The inner summation over $r$ can be easily seen to converge to a constant, since the exponent can be made
negative for suitably large $A$, and almost doubles with increase in $r$.  The summation over $i$, converges to $\Theta(n)$, 
and can be made much smaller than $n/32$, by suitably adjusting the constant $B$.

For the second part, the calculations proceed as below: \\
For $j\leq j_0$, $i=1..k$, just upper bound the exponent by $1$. Now we have:
$$ \sum_{i=1}^k \sum_{j=i-1}^{j_0} C.\frac{2^{jd}}{2^{dB}2^{d_2(i-1)}} .$$
Reverse the order of summation:
$$ \sum_{j=0}^{\log n/d} \sum_{i=j+1}^{i_0} C.\frac{2^{jd}}{2^{dB}2^{d_2(i-1)}} ,$$
where $i_0=1+\frac{\log n-jd}{d_2}.$ This further simplifies to:
\begin{eqnarray*}
\sum_{j=0}^{\log n/d} \sum_{i=j+1}^{i_0} C.\frac{2^{jd}}{2^{dB}2^{d_2(i-1)}} 
   &\leq& 
\sum_{j=0}^{\log n/d} 2C.\frac{2^{jd}}{2^{d_2(j)}2^{dB}} \\
   &=& 
\sum_{j=0}^{\log n/d} 2C.2^{j(d-d_2)-dB} \\
   &\leq& 
2C.2^{\log n(d_1/d)-dB} = n^{d_1/d}/2^{dB}
\end{eqnarray*}
which is much less than $n$, for suitable value of $B$.

To get the discrepancy bound for the original set $S$, which had size in $[n/2^{i-1},n/2^i)$, we need to sum 
up the discrepancies over the $F_j^i$ in the chain corresponding to $S$: $\Delta_S=\sum_{j}\Delta_j^i$. 
Here, the factor $\frac{A}{(1+|j-j_0|)^2}$ in our 
choice of $\Delta_j^i$ ensures that this sum is essentially of the order of the maximal term, which occurs when 
$j=j_0$. Further, when $d_1=1$, 
the $\log n$ rounds of the Lovett-Meka algorithm induce an extra logarithmic factor. In the case $d_1>1$ 
this does not happen, because the factor of $n^{(d_1-1)/(2d)}$ present in the discrepancy bound sets up 
a geometrically decreasing series.

Therefore, in terms of the size of the original set $S$, which was in $[n/2^{i-1},n/2^i)$, we get:
\begin{eqnarray*}
\Delta_S \leq \sum_{j=i-1}^k \Delta_j^i &=& \sum_{i-1}^k A\frac{n^{1/2-1/(2d)}}{2^{(i-1)(d_2/(2d))}}\sqrt{1+2\log(k/2-|i-k/2|)}  \\
                       &=& \left\{ \begin{array}{ll}
                       O\left(|S|^{d_2/(2d)}n^{(d_1-1)/(2d)}\sqrt{\log f}\right), & \mbox{ if } d_1>1 \\
                       O\left(|S|^{d_2/(2d)}\log n\sqrt{\log f}\right), & \mbox{ if } d_1=1 
                       \end{array} \right.
\end{eqnarray*}
where $$f = \left\{\begin{array}{ll}
           1+2\log|S|, & \mbox{ if } |S| \leq n^{1/2} \\
           1+2\log (n/|S|), & \mbox{ if } |S| \geq n^{1/2}
           \end{array}
           \right.
$$

\vspace{10pt}

\noindent
This completes the proof of Theorem~\ref{thm-appln-disc-bds}.


\section{`Well-behaved' Beck-Fiala systems}
\label{sec-well-behaved-Beck-Fiala}

In many geometric settings, the range space is such that for any subset of elements in the universe,
the number of projections of the range space on this subset is linear in the size of the subset. For 
example, points and axis-parallel rectangles in the plane, or points and axis-parallel boxes in three dimensions,
points and half-spaces in two and three dimensions, etc. In general, following Ezra \cite{Ezra-sizesendisc-soda-14}, we 
call a set-system `well-behaved' if it has bounded primal shatter dimension, and $d_1=1$ for this system. In this section, 
we shall prove a general result for such systems, under the Beck-Fiala setting.


\paragraph*{Details}
 
In the Beck-Fiala setting, each element $x \in X$ has degree bounded by $t$, i.e. belongs to at most $t$ 
many ranges or sets. Suppose, in addition, the range space is also planar, i.e. the ranges appear 
as polygons on a plane, then we obtain the following result:

\begin{theo}\label{coro-well-behaved-Beck-Fiala}
   Let $(X,\mathcal{S})$ be a (finite) set system with bounded primal shatter dimension $d$, and size-sensitive constants 
$d_1=1$, and $d_2=d-d_1$. Further, each element belongs to at most $t$ sets. Then the discrepancy of this set system is given by:
$$ 
\mathrm{disc}(\mathcal{S}) = O(t^{1/2-1/2d}\sqrt{\log\log t}\log n ).
$$
Note that this discrepancy bound is constructive, i.e., in expected
polynomial time we can find a coloring that matches the above
discrepancy bound. 
\end{theo}

\paragraph{Proof Sketch:} The proof follows quite simply. First, we observe 
that in a Beck-Fiala-type system with maximum degree $t$,
the number of sets having size more than $32t$ is less than $n/32$. 
If we ensure that $\sum_{S:|S|\leq 32t} exp(-\Delta_S^2/16|S|)$ 
is at most $n/32$, then each of the remaining sets $S$ can be assigned zero discrepancy (i.e. $\forall S: |S|\geq 32t:\Delta_S=0$)
throughout the $O(\log n)$ rounds of the algorithm.
We apply Theorem~\ref{thm-appln-disc-bds} 
for $d_1=1$ only for the sets whose size is at most $32t$, and set 
$\Delta_S=0$ for each set that has more than $32t$ vertices. Thus we get 
that the maximum discrepancy is 
$O(t^{1/2-1/2d}\sqrt{\log \log t}\log n)$.

\vspace{5pt}
\noindent {\bf Example.} Observe that
Theorem~\ref{coro-well-behaved-Beck-Fiala} implies that for the case of
points and halfspaces in the $2$-dimensional case, in the Beck-Fiala
setting, the discrepancy is bounded by 
$$
O(t^{1/4} \sqrt{\log\log t}\log n).
$$

\section{Improving relative $(\varepsilon,\delta)$-approximation bound via discrepancy}
\label{sec-improving-sample-bound-via-discrepancy}

Har-Peled and
Sharir~\cite{Har-PeledS11-relative-approximation-geometry} 
showed that the notion of relative $(\varepsilon, \delta)$-approximation
and $(\nu, \alpha)$-sample are equivalent if $\nu$ is proportional $\varepsilon$ and $\alpha$
is proportional to $\delta$. 
In this section we will be working with $(\nu, \alpha)$-sample, and
the improvements in the bounds we get in $(\nu,\alpha)$-sample size
will directly imply improvement in the size of relative 
$(\varepsilon, \delta)$-approximation.


We will use the {\em halfing technique}~\cite{MatousekWW93-discrepancy-approx-VC,
Har-PeledS11-relative-approximation-geometry,Ezra-sizesendisc-soda-14} 
repeatedly to get a $(\nu,\alpha)$-sample. In the analysis of this
procedure the size sensitve discrepancy bound for the set system $(X, \mathcal{S})$ plays
an important role. The reason we could improve on the previous bounds
of $(\nu, \alpha)$-sample because we improved on the size sensitive 
discrepancy bounds in Theorem~\ref{thm-appln-disc-bds}.  

In this construction the set $X$ is repeatdly halved in each
iterations 
until one obtains a $(\nu,\alpha)$-sample of appropriate size.
W.l.o.g we will assume for the rest of this section that $X \in \mathcal{S}$.
Let $X_{0} = X$. In each iteration $i \geq 1$, the set $X_{i-1}$ is
partitioned into sets $X_{i-1}$ and $X'_{i-1}$ where $X_{i-1}$ is colored
$+1$ and $X'_{i-1}$ is colored $-1$. Note that the coloring corresponds
to the discrepancy bounds obtained in Theorem~\ref{thm-appln-disc-bds}.
Assume, w.l.o.g., $|X_{i-1}| \geq |{X_{i-1}}'|$.
We continue this process until a $(\nu, \alpha)$-sample of desired
size is obtained.

\paragraph*{Case of $d_{1} > 1$}

Since we assume $X_{0} = X \in \mathcal{S}$, we would get for each
iteration $i$, $X_{i-1} \in \mathcal{S}\mid_{X_{i-1}}$.
From Theorem~\ref{thm-appln-disc-bds}, we get, for all $S \in \mathcal{S}$, that
\begin{eqnarray}
  | |X_{i} \cap S| - |X_{i}'\cap S| | \leq K_{d}
  |X_{i-1}\cap S|^{\frac{d-d_{1}}{2d}} |X_{i-1}|^{\frac{d_{1}-1}{2d}}
  \, f(|X_{i-1} \cap S|, |X_{i-1}|).
\end{eqnarray}
The constant $K_{d}$ depends only on $d$ and the function 
$f(,)$ is defined in Theorem~\ref{thm-appln-disc-bds}.

Taking $S = X_{i-1}$ and using the fact $|X_{i}| + |X_{i}'| =
|X_{i-1}|$ we get 
\begin{eqnarray}
  | |X_{i}| - (|X_{i-1}| - |X_{i}| ) |
  &\leq & K_{d} |X_{i-1}|^{\frac{d-d_{1}}{2d}}
  |X_{i-1}|^{\frac{d_{1}-1}{2d}} 
  \quad\mbox{as $f(|X_{i-1}|, |X_{i-1}|) =1$}\nonumber \\
  &=& K_{d} |X_{i-1}|^{\frac{d-1}{2d}}
\end{eqnarray}
Therefore
\begin{eqnarray}
  \left| |X_{i}| - \frac{|X_{i-1}|}{2}\right| \leq \frac{K_{d}}{2} |X_{i-1}|^{\frac{d-1}{2d}}. 
\end{eqnarray}
This implies 
\begin{eqnarray*}
  |X_{i}| &\leq& \frac{|X_{i-1}|}{2} \left( 1 + \frac{K_{d} \,
      |X_{i}|^{\frac{d-1}{2d}}}{|X_{i-1}|} \right)\\
  &=& \frac{|X_{i-1}|}{2} \left( 1 + \frac{K_{d}}{|X_{i-1}|^{\frac{d+1}{2d}}} \right)
\end{eqnarray*}

Write $|X_{i}|$ as 
\begin{eqnarray*}
  |X_{i}| = \frac{|X_{i-1}|}{2} (1 + \delta_{i-1})
  \;\;\mbox{where}\;\; 0 \leq \delta_{i-1} \leq
  \frac{K_{d}}{|X_{i-1}|^{\frac{d+1}{2d}}}.
\end{eqnarray*}
Above inductive formula gives the following recursive formula:
\begin{eqnarray}\label{equation-bounding-X-i}
  |X_{i}| &=& \frac{|X_{0}|}{2^{i}} \prod_{j=0}^{i-1}(1+\delta_{j}) \nonumber\\
  &\leq& \frac{|X_{0}|}{2^{i}} \exp\left\{ \sum_{j=0}^{i-1} \delta_{j}
  \right\}\nonumber\\ 
  &\leq&  \frac{n}{2^{i}} \exp\left\{ K_{d} \sum_{j=0}^{i-1} \left( \frac{2^{j}}{n}\right)^{1/2+1/2d}
  \right\} 
\end{eqnarray}
The last inequality follows from the fact $|X_{i}| \geq |X_{i-1}|/2$
(by construction)
and $\delta_{j} \leq \frac{K_{d}}{|X_{i-1}|^{\frac{d+1}{2d}}}$.
The exponential term in Eq.~\eqref{equation-bounding-X-i} to
$O(1)$ if 
$$
i \leq \log n - \frac{2d}{d+1} \log K_{d}.
$$ 
Stopping the procedure for the above mentioned bound we get $n_{i} = \Omega(K_{d}^{\frac{2d}{d+1}})$.


Observe that 
\begin{align}\label{eqn-sample-size-1}
  \left|\overline{X}_{i-1}(S) - \overline{X}_{i}(S) \right| &= \left| \frac{|S\cap
      X_{i-1}|}{|X_{i-1}|} - \frac{|S\cap X_{i}|}{|X_{i}|}\right|&\nonumber\\
  &= \left| \frac{|S\cap X_{i}|+|S\cap X_{i-1}|}{|X_{i-1}|} -
    \frac{|S\cap X_{i}|}{|X_{i}|}\right|& \mbox{as $X_{i-1} = X_{i}
    \sqcup X_{i}'$} \nonumber\\
  &= \left| \frac{|S\cap X_{i}|+|S\cap X_{i-1}|}{|X_{i-1}|} -
    \frac{2|S\cap X_{i}|}{|X_{i-1}| (1+ \delta_{i-1})}\right|& \mbox{as $|X_{i}| =
    \frac{|X_{i-1}|(1+\delta_{i-1})}{2}$}\nonumber\\
  &= \left| \frac{|S\cap X'_{i}|}{|X_{i-1}|} - \frac{|S\cap
      X_{i}|(1-\delta_{i-1})}{|X_{i-1}|(1+\delta_{i-1})}\right|&\nonumber\\
  &=\left| \frac{|S\cap X_{i}'| - |S\cap X_{i}|}{|X_{i-1}|} + \frac{2\delta_{i-1}}{1+\delta_{i-1}}
  \frac{|S\cap X_{i}|}{|X_{i-1}|}\right| & \mbox{add. \&
  sub. $\frac{|S\cap X_{i}|}{|X_{i-1}|}$}\nonumber\\
  &=\left| \frac{|S\cap X_{i}'| - |S\cap X_{i}|}{|X_{i-1}|} + \delta_{i-1}
  \frac{|S\cap X_{i}|}{|X_{i}|}\right| & \mbox{as $|X_{i-1}| =
  \frac{2|X_{i}|}{(1+\delta_{i-1})}$}\nonumber\\
  &=\left| \frac{|S\cap X_{i}'| - |S\cap X_{i}|}{|X_{i-1}|}\right| +
  \delta_{i-1} \overline{X}_{i}(S)&
\end{align}
Note that the second term in the above equation is bounded by 
$$
\delta_{i-1} \overline{X}_{i}(S) \leq \delta_{i-1} \left(
  \overline{X}_{i}(S) + \overline{X}_{i-1}(S) + \nu\right)
$$
There now we would try to bound the first term $\left| \frac{|S\cap
    X_{i}'| - |S\cap X_{i}|}{|X_{i-1}|}\right|$.

From Theorem~\ref{thm-appln-disc-bds}, the 
fact that $f(|X_{i-1}\cap S|, |X_{i-1}|) = O(\log \log |X_{i-1}|)$,
and $x^{t} \leq \frac{(x+y)}{y^{1-t}}$, $\forall x \geq
    0, \, y > 0$, \& $ t \in [0,\, 1]$, 
we get that there exists $K'_{d}$ such that  
\begin{align}\label{eqn-sample-size-2}
  \left| \frac{|S\cap X_{i}'| - |S\cap X_{i}|}{|X_{i-1}|}\right| &\leq
  \frac{K'_{d}
    \overline{X}_{i-1}(S)^{\frac{d_{2}} {2d} } \log \log |X_{i-1}| }{|X_{i-1}|^{\frac{d+1}{2d}}}&\nonumber\\
  &\leq  \frac{K'_{d} \, (\overline{X}_{i-1}(S) +
    \nu) \, \log \log |X_{i-1}|}{|X_{i-1}|^{\frac{d+1}{2d}} \; \nu^{\frac{d+ d_{1}}{2d} }}&
  \nonumber\\
  &\leq  \frac{K'_{d} \log \log |X_{i-1}|}{|X_{i-1}|^{\frac{d+1}{2d}}}
  \frac{\overline{X}_{i}(S)+\overline{X}_{i-1}(S) +
    \nu}{\nu^{\frac{d+d_{1}}{2d}}}& 
\end{align}

From Eq.s~\eqref{eqn-sample-size-1} and \eqref{eqn-sample-size-2}, we get 
\begin{eqnarray*}
  d_{\nu} (\overline{X}_{i-1}(S), \overline{X}_{i}(S)) &\leq&
  \frac{K'_{d}}{|X_{i-1}|^{\frac{d+1}{2d}}} \left( 1+
    \frac{\log \log |X_{i}|}{\nu^{\frac{d+d_{1}}{2d}}}\right)\\
  &\leq& 
  \frac{2K'_{d} \log \log |X_{i-1}| }{|X_{i-1}|^{\frac{d+1}{2d}}
    \nu^{\frac{d+d_{1}}{2d}}} 
\end{eqnarray*}

Using the fact that $d_{\nu}(\cdot, \cdot)$ satisfies triangle
inequality~\cite{Haussler-decision-theoretic-PAC-learning-92,LiLS-sample-complexity-learning-S01}, 
we get 
\begin{eqnarray}
  d_{\nu}(\overline{X}_{0}(S), \overline{X}_{i}(S)) &\leq&
  \sum_{j=1}^{i} d_{\nu}(\overline{X}_{j-1}(S),
  \overline{X}_{j}(S)) \nonumber\\
  &\leq& O\left(\frac{\log \log n_{i} }{\nu^{\frac{d+d_{1}}{2d}}
      n_{i-1}^{\frac{d+1}{2d} }} \right),
\end{eqnarray}
the constant in big-$O$ depends only on $d$.

This implies to get $d_{\nu}(\overline{X}_{0}(S), \overline{X}_{i}(S))
< \alpha$, we need
$$
n_{i-1} = \Omega\left( \frac{\log \log^{\frac{2d}{d+1}}
    \frac{1}{\nu\alpha}}{\nu^{\frac{d+d_{1}}{d+1}} \alpha^{\frac{2d}{d+1}}}\right).
$$

Therefore there exists a $(\nu,\alpha)$-sample of size 
$$
O\left( \frac{\log \log^{\frac{2d}{d+1}}
    \frac{1}{\nu\alpha}}{\nu^{\frac{d+d_{1}}{d+1}}
    \alpha^{\frac{2d}{d+1}}}\right).
$$

\paragraph*{Case of $d_{1} = 1$}

Using the same technique as for the case of $d_{1} > 1$, we will get
the following bound for $(\nu, \alpha)$-sample size 
$$
\max \left\{ O\left(\log^{\frac{2d}{d+1}} n\right), 
  \, O\left( \frac{\log^{\frac{2d}{d+1}} n \, \log
      \log^{\frac{2d}{d+1}} n}{\nu
      \alpha^{\frac{2d}{d+1}}}\right)\right\} .
$$
Note that the constant in big-$O$ depends only on $d$.

\vspace{10pt}

This completes the proof of Theorem~\ref{thm-main-sample-size-bound}.



\bibliography{Packing}

\begin{thebibliography}{10}

\bibitem{Banaszczyk-rsa-98}
W.~Banaszczyk.
\newblock Balancing vectors and gaussian measures of n-dimensional convex
  bodies.
\newblock {\em Random Struct. Algorithms}, 12(4):351--360, 1998.

\bibitem{Bansal-discrepancy-minimization-focs-10}
N.~Bansal.
\newblock Constructive {A}lgorithms for {D}iscrepancy {M}inimization.
\newblock In {\em 51th Annual {IEEE} Symposium on Foundations of Computer
  Science, {FOCS} 2010, October 23-26, 2010, Las Vegas, Nevada, {USA}}, pages
  3--10, 2010.

\bibitem{Beck-discrepancy-integer-sequence-81}
J.~Beck.
\newblock {R}oths estimates on the discrepancy of integer sequences is nearly
  sharp.
\newblock {\em Combinatorica}, 1(4):319--325, 1981.

\bibitem{Spencer-six-standard-85}
J.~Beck.
\newblock Six standard deviations suffice.
\newblock {\em Trans. Amer. Math. Soc.}, 289(2):679--706, 1985.

\bibitem{Beck-Fiala-integer-making-1981}
J.~Beck and T.~Fialq.
\newblock ``integer making'' theorems.
\newblock {\em Discrete Applied Math.}, 3:1--8, 1981.

\bibitem{Chazelle-packing-92}
B.~Chazelle.
\newblock A note on {H}aussler's packing lemma.
\newblock Technical report, Princeton, 1992.

\bibitem{Haussler92spherepacking}
D.~David~Haussler.
\newblock Sphere {P}acking {N}umbers for {S}ubsets of the {B}oolean n-{C}ube
  with {B}ounded {V}apnik-{C}hervonenkis {D}imension.
\newblock {\em J. Comb. Theory, Ser. {A}}, 69(2):217--232, 1995.

\bibitem{Ezra-smallsizeapprangespace-socg-13}
E.~Ezra.
\newblock Small-size relative $(p, \varepsilon)$-approximations for
  well-behaved range spaces.
\newblock In {\em Symposium on Computational Geometry}, pages 233--242, 2013.

\bibitem{Ezra-sizesendisc-soda-14}
E~Ezra.
\newblock A {S}ize-{S}ensitive {D}iscrepancy {B}ound for {S}et {S}ystems of
  {B}ounded {P}rimal {S}hatter {D}imension.
\newblock In {\em SODA}, pages 1378--1388, 2014.

\bibitem{Har-PeledS11-relative-approximation-geometry}
S.~Har{-}Peled and M.~Sharir.
\newblock Relative $(\emph{p}, \emph{{\(\epsilon\)}})$-{A}pproximations in
  {G}eometry.
\newblock {\em Discrete {\&} Computational Geometry}, 45(3):462--496, 2011.

\bibitem{Haussler-decision-theoretic-PAC-learning-92}
D.~Haussler.
\newblock Decision {T}heoretic {G}eneralizations of the {PAC} {M}odel for
  {N}eural {N}et and {O}ther {L}earning {A}pplications.
\newblock {\em Inf. Comput.}, 100(1):78--150, 1992.

\bibitem{LiLS-sample-complexity-learning-S01}
Y.~Li, P.~M. Long, and A.~Srinivasan.
\newblock Improved {B}ounds on the {S}ample {C}omplexity of {L}earning.
\newblock {\em J. Comput. Syst. Sci.}, 62(3):516--527, 2001.

\bibitem{Lovett-Meka-discmin-focs-12}
S.~Lovett and R.~Meka.
\newblock Constructive {D}iscrepancy {M}inimization by {W}alking on the
  {E}dges.
\newblock In {\em FOCS}, pages 61--67, 2012.

\bibitem{Matousek-tight-half-spaces-95}
J.~Matousek.
\newblock Tight {U}pper {B}ounds for the {D}iscrepancy of {H}alf-{S}paces.
\newblock {\em Discrete {\&} Computational Geometry}, 13:593--601, 1995.

\bibitem{matousek-geomdisc-2009}
J.~Matousek.
\newblock {\em Geometric {D}iscrepancy: {A}n {I}lustrated {Guide} ({A}lgorithms
  and {C}ombinatorics)}.
\newblock Springer, 1999.

\bibitem{MatousekWW93-discrepancy-approx-VC}
J.~Matousek, E.~Welzl, and L.~Wernisch.
\newblock Discrepancy and approximations for bounded {V}{C}-dimension.
\newblock {\em Combinatorica}, 13(4):455--466, 1993.

\bibitem{SharirZ-range-searching-13}
M.~Sharir and S.~Zaban.
\newblock Output-{S}ensitive {T}ools for {R}ange {S}earching in {H}igher
  {D}imensions.
\newblock {\em CoRR}, abs/1312.6305, 2013.

\bibitem{Srinivasan-soda-Matrices97}
A.~Srinivasan.
\newblock Improving the {D}iscrepancy {B}ound for {S}parse {M}atrices: {B}etter
  {A}pproximations for {S}parse {L}attice {A}pproximation {P}roblems.
\newblock In {\em Proceedings of the Eighth Annual {ACM-SIAM} Symposium on
  Discrete Algorithms, 5-7 January 1997, New Orleans, Louisiana.}, pages
  692--701, 1997.

\bibitem{Wernisch-packing-92}
L.~Wernisch.
\newblock Manuscript.
\newblock Technical report, FU Berlin, 1992.

\end{thebibliography}

\end{document}